\documentclass[11pt,letterpaper]{article}

\usepackage{url}
\usepackage{fullpage}
\usepackage{cite}
\usepackage{xspace}

\usepackage[subtle,mathspacing=normal]{savetrees}
\usepackage{thmtools}
\usepackage{thm-restate}

\usepackage{amsmath}
\usepackage{algorithm}
\usepackage[noend]{algpseudocode}
\usepackage{amssymb}
\usepackage{amsthm}
\usepackage{color}
\usepackage{array}
\usepackage{xy}
\usepackage{setspace}
\usepackage{varwidth}
\usepackage{multicol}
\usepackage{algorithm}
\usepackage{hyperref}
\usepackage{todonotes}

\newcommand{\NN}{\mathbb{N}} 

\newcommand{\poly}{\operatorname{poly}}

 \newcommand{\E}{\mathbb{E}}

\usepackage{xpatch}
\xpretocmd{\eqref}{Eq.~}{}{}

\newcommand{\elastic}{elastic\xspace}
\newcommand{\Elastic}{Elastic\xspace}

\newcommand{\probecomplexity}{probe complexity\xspace}

\newcommand{\ProbeComplexity}{Probe Complexity\xspace}

\newtheorem{thm}{Theorem}

\theoremstyle{remark}

\newtheorem{theorem}{Theorem}

\newtheorem{lemma}[thm]{Lemma}
\newtheorem{proposition}[thm]{Proposition}
\newtheorem{claim}[thm]{Claim}

\theoremstyle{remark}

\newcommand{\defn}[1]{\textbf{\emph{#1}}}
\renewcommand{\paragraph}[1]{\vspace{.2 cm} \noindent \textbf{#1}}

\usepackage{authblk}

\begin{document}

\title{Optimal Bounds for Open Addressing Without Reordering}
\author{Mart\'in Farach-Colton\footnote{NYU. \texttt{martin@farach-colton.com}}, Andrew Krapivin\footnote{Cambridge University. \texttt{andrew@krapivin.net}}, William Kuszmaul\footnote{CMU. \texttt{kuszmaul@cmu.edu}}}
\date{}
\maketitle
\begin{abstract}
   In this paper, we revisit one of the simplest problems in data structures: the task of inserting elements into an open-addressed hash table so that elements can later be retrieved with as few probes as possible. We show that, even without reordering elements over time, it is possible to construct a hash table that achieves far better expected search complexities (both amortized and worst-case) than were previously thought possible. Along the way, we disprove the central conjecture left by Yao in his seminal paper \emph{``Uniform Hashing is Optimal''}. All of our results come with matching lower bounds.
\end{abstract}
\section{Introduction}\label{sec:intro}
In this paper, we revisit one of the simplest problems in data structures: the task of inserting elements into an open-addressed hash table so that elements can later be retrieved with as few probes as possible. We show that, even without reordering elements over time, it is possible to construct a hash table that achieves far better expected probe complexities (both amortized and worst-case) than were previously thought possible. Along the way, we disprove the central conjecture left by Yao in his seminal paper \emph{``Uniform Hashing is Optimal''} \cite{yao1985uniform}.

\paragraph{Background.}
Consider the following basic problem of constructing an \defn{open-addressed hash table without reordering}. A sequence $x_1, x_2, \ldots, x_{(1 - \delta)n}$ of keys are inserted, one after another, into an array of size $n$. Each $x_i$ comes with a \defn{probe sequence} $h_1(x_i), h_2(x_i), \ldots \in [n]^{\infty}$ drawn independently from some distribution $\mathcal{P}$. To insert an element $x_i$, an \defn{insertion algorithm} $\mathcal{A}$ must choose some not-yet-occupied position $h_j(x)$ in which to place the element. Note that insertions cannot reorder (i.e., move around) the elements that were inserted in the past, so the only job of an insertion is to select which unoccupied slot to use. The full specification of the hash table is given by the pair $(\mathcal{P}, \mathcal{A})$. 

If $x_i$ is placed in position $h_j(x_i)$, then $x_i$ is said to have \defn{\probecomplexity} $j$. This refers to the fact that a query could find $x_i$ by making $j$ probes to positions $h_1(x), \ldots h_j(x)$. The goal is to design the hash table $(\mathcal{P}, \mathcal{A})$ in a way that minimizes the \defn{amortized expected \probecomplexity}, i.e., the expected value of the average \probecomplexity across all of the keys $x_1, x_2, \ldots, x_{(1 - \delta) n}$.

The classic solution to this problem is to use \defn{uniform probing} \cite{Knuth98Vol3}: the probe sequence for each key is a random permutation of $\{1, 2, \ldots, n\}$, and each insertion $x_i$ greedily uses the first unoccupied position from its probe sequence. It is a straightforward calculation to see that random probing has amortized expected \probecomplexity $\Theta(\log \delta^{-1})$. 

Ullman conjectured in 1972 \cite{ullman1972note} that the amortized expected \probecomplexity of $\Theta(\log \delta^{-1})$ should be optimal across all \defn{greedy} algorithms, i.e., any algorithm in which each element uses the first unoccupied position in its probe sequence. This conjecture remained open for more than a decade before it was proven by Yao in 1985 \cite{yao1985uniform} in a celebrated paper titled ``\emph{Uniform Hashing is Optimal}''.

The classical way to get around Yao's lower bound is to consider a relaxation of the problem in which the insertion algorithm is permitted to perform \emph{reordering}, i.e., moving elements around after they're inserted. In this relaxed setting, it is possible to achieve $O(1)$ amortized expected probe complexity even when the hash table is completely full \cite{brent1973reducing, gonnet1979efficient, munro1986techniques}. What is not clear is whether this relaxation is necessary. Could a non-greedy algorithm potentially achieve $o(\log \delta^{-1})$ amortized expected probe complexity, without reordering? Or is reordering fundamentally necessary to achieve small amortized probe complexity?

\vspace{.4 cm}

\hspace{.4 cm} \fbox{%
    \parbox{0.9 \textwidth}{%
       \textbf{Question 1. }Can an open-addressed hash table achieve amortized expected probe complexity $o(\log \delta^{-1})$ without reordering elements after they are inserted?
    }%
}

\vspace{.4 cm}

A closely related problem is that of minimizing \defn{worst-case expected \probecomplexity}. A \emph{worst-case} bound on expected \probecomplexity must apply to each insertion individually---even to the insertions that are performed when the hash table is very full. Uniform probing achieves a worst-case expected \probecomplexity of $O(\delta^{-1})$. It has remained an open question, however, whether this bound is asymptotically optimal without the use of reordering.


\vspace{.4 cm}

\hspace{.4 cm} \fbox{%
    \parbox{0.9 \textwidth}{%
       \textbf{Question 2. }Can an open-addressed hash table achieve worst-case expected probe complexity $o(\delta^{-1})$ without reordering?
    }%
}

\vspace{.4 cm}

This second question, somewhat notoriously \cite{ullman1972note, yao1985uniform, martini2003double, burkhard2005external, molodowitch1990analysis, lueker1988more}, remains open even for \emph{greedy} open-addressed hash tables. Yao conjectured in 1985 \cite{yao1985uniform} that uniform probing should be nearly optimal in this setting, that is, that any greedy open-addressed hash table must have worst-case expected \probecomplexity at least $(1 - o(1)) \delta^{-1}$. Despite its simplicity, Yao's conjecture has never been settled.

\vspace{.4 cm}

\paragraph{This paper: Tight bounds for open addressing without reordering.} In Section~\ref{sec:elastic}, we give a single hash table that answers both of the above questions in the affirmative. Specifically, we show how to achieve an amortized bound of $O(1)$ and a worst-case bound of $O(\log \delta^{-1})$ on the expected \probecomplexity in an open-addressed hash table that does not make use of reordering.

\begin{restatable}{theorem}{thmmain}
Let $n \in \mathbb{N}$ and $\delta \in (0, 1)$ be parameters such that $\delta > O(1/n)$ and $\delta^{-1}$ is a power of two. It is possible to construct an open-addressing hash table that supports $n - \lfloor \delta n \rfloor$ insertions in an array of size $n$, that does not reorder items after they are inserted, and that offers amortized expected probe complexity $O(1)$, worst-case expected \probecomplexity $O(\log \delta^{-1})$, and worst-case expected insertion time $O(\log \delta^{-1})$.
\label{thm:main-intro}
\end{restatable}

We refer to our insertion strategy as \defn{\elastic hashing}, because of the way that the hash table often probes much further down the probe sequence before snapping back to the position it ends up using. That is, in the process of deciding which slot $h_i(x)$ to put a key $x$ in, the algorithm will first examine many slots $h_j(x)$ satisfying $j > i$. This non-greedy behavior is essential, as it is the only possible way that one could hope to avoid Yao's lower bound \cite{yao1985uniform} without reordering. 

Our bound of $O(1)$ on amortized expected \probecomplexity is, of course, optimal. But what about the bound of $O(\log \delta^{-1})$ on worst-case expected \probecomplexity? We prove that this bound is also optimal: any open-addressing hash table that does not use reordering must have worst-case expected \probecomplexity at least $\Omega(\log \delta^{-1})$. 

Next, in Section \ref{sec:funnel}, we turn our attention to \emph{greedy} open-addressed hash tables. Recall that, in this setting, Question 1 has already been resolved -- it has been known for decades that uniform probing is asymptotically optimal \cite{yao1985uniform}. Question 2, on the other hand, remains open -- this is the setting where Yao conjectured uniform probing to be optimal \cite{ullman1972note, yao1985uniform, martini2003double, burkhard2005external, molodowitch1990analysis, lueker1988more}. 
Our second result is a simple greedy open-addressed strategy, which we call \defn{funnel hashing}, that achieves $O(\log^2 \delta^{-1})$ worst-case expected \probecomplexity:

\begin{restatable}{theorem}{thmgreedy}
 Let $n \in \mathbb{N}$ and $\delta \in (0, 1)$ be parameters such that $\delta > O(1/n^{o(1)})$. There is a greedy open-addressing strategy that supports $n - \lfloor \delta n \rfloor$ insertions in an array of size $n$, and that offers worst-case expected \probecomplexity (and insertion time)  
  $O(\log^2 \delta^{-1}).$ 
 Furthermore, the strategy guarantees that, with probability $1 - 1 / \poly(n)$, the worst-case \probecomplexity over all insertions is 
 $O(\log^2 \delta^{-1} + \log \log n)$.
 Finally, the amortized expected \probecomplexity is $O(\log \delta^{-1})$.
 \label{thm:greedy-intro}
\end{restatable}

The bound of $O(\log^2 \delta^{-1}) = o(\delta^{-1})$ on worst-case expected \probecomplexity is asymptotically smaller than the $\Theta(\delta^{-1})$ bound that Yao conjectured to be optimal. This means that Yao's conjecture is false, and that there is a sense in which uniform probing is \emph{sub-optimal} even among greedy open-addressed strategies.

Despite their somewhat unusual shape, the bounds in Theorem \ref{thm:greedy-intro} turn out to be optimal. For worst-case expected \probecomplexity, we prove a matching lower bound of $\Omega(\log^2 \delta^{-1})$ that applies to any greedy open-addressing scheme. For high-probability worst-case \probecomplexity, we prove a matching lower bound $\Omega(\log^2 \delta^{-1} + \log \log n)$ that applies to any open-addressing algorithm that does not perform reorderings. Perhaps surprisingly, this second lower bound holds even to non-greedy strategies.

The basic structure of funnel hashing, the hash table that we use to prove Theorem \ref{thm:greedy-intro}, is quite simple, and subsequent to the initial version of this paper, the authors have also learned of several other hash tables that make use of the same high-level idea in different settings \cite{broder1990multilevel, fotakis2005space}. Multi-level adaptive hashing \cite{broder1990multilevel} uses a similar structure (but only at low load factors) to obtain a hash table with $O(\log \log n)$ levels that supports high parallelism in its queries -- this idea was also applied subsequently to the design of contention-resolution schemes \cite{contention}. Filter hashing \cite{fotakis2005space} applied the structure to high load factors to get an alternative to $d$-ary cuckoo hashing that, unlike known analyses of standard $d$-ary cuckoo hashing, can be implemented with constant-time polynomial hash functions. An alternative path to disproving Yao's conjecture would be to directly modify filter hashing to use a greedy open-addressed hash table (e.g., linear probing) in its final layer.

\paragraph{Additional problem history and related work. }We conclude the introduction by briefly giving some additional discussion of related work and of the history of the problems and models studied in this paper.

The idea of studying amortized expected \probecomplexity appears to have been first due to Knuth in his 1963 paper on linear probing \cite{knuth1963notes}. Knuth observed that, when a linear-probing hash table is $1 - \delta$ full, then even though the expected insertion time is $O(\delta^{-2})$, the amortized expected \probecomplexity is $O(\delta^{-1})$. Knuth would later pose a weaker version of Ullman's conjecture \cite{knuth1974computer}, namely that uniform probing is optimal out of a restricted set of greedy strategies known as \emph{single-hashing} strategies. This weaker conjecture was subsequently proven by Ajtai \cite{ajtai1978there}, whose techniques ended up serving as the eventual basis for Yao's proof of the full conjecture \cite{yao1985uniform}. As noted earlier, Yao conjectured that it should be possible to obtain a stronger result, namely that the \emph{worst-case} expected probe complexity in any greedy open-addressing hash table is $\Omega(\delta^{-1})$. This conjecture remained open \cite{martini2003double, burkhard2005external, molodowitch1990analysis, lueker1988more} until the current paper, which disproves it in Theorem \ref{thm:greedy-intro}.


Although we do not discuss key-value pairs in this paper, most applications of open-addressing associate a value with each key \cite{Knuth98Vol3}. In these settings, the job of a query is not necessarily to determine whether the key is present (it is often already known to be), but instead to recover the corresponding value. This distinction is important because both probe complexity and amortized probe complexity are notions that apply only to keys that \emph{are present}. Minimizing amortized expected probe complexity, in particular, corresponds to minimizing the expected time to query a \emph{random element out of those present}. There is no similar notion for negative queries---when querying an element not present, there is no interesting difference between querying an arbitrary element versus a random one. 

For \emph{worst-case} expected probe complexity, on the other hand, one can in some cases hope to extend one's results to negative queries. For greedy algorithms, in particular, negative query time is the same as insertion time (both stop when they encounter a free slot) \cite{Knuth98Vol3}. Thus the guarantee in Theorem \ref{thm:greedy-intro} extends to imply an $O(\log^2 \delta^{-1})$ expected time bound for negative queries.

One can also extend the study of open-addressing without reordering to settings that support both insertions and deletions over an infinite time horizon \cite{sanders2018hashing, bender2023tinypointers, bender2023iceberg}. In this setting, even very basic schemes such as linear probing \cite{sanders2018hashing} and uniform probing \cite{bender2023tinypointers} have resisted analysis---it is not known whether either scheme achieves expected insertion times, probe complexities, or amortized probe complexities that even bounded as a function of $\delta^{-1}$. It is known, however, that the optimal amortized expected probe complexity in this setting is $\delta^{-\Omega(1)}$ (see Theorem 3 in \cite{bender2023tinypointers}), meaning that results such as Theorems \ref{thm:main-intro} and \ref{thm:greedy-intro} are not possible.

\section{\Elastic Hashing}\label{sec:elastic}

In this section, we construct elastic hashing, an open-addressed hash table (without reordering) that achieves $O(1)$ amortized expected probe complexity and $O(\log \delta^{-1})$ worst-case expected probe complexity.

\thmmain*

Our construction will make use of a specific injection $\phi: \mathbb{Z}^+ \times \mathbb{Z}^+ \rightarrow \mathbb{Z}^+$.

\begin{lemma}
    There exists an injection $\phi: \mathbb{Z}^+ \times \mathbb{Z}^+ \rightarrow \mathbb{Z}^+$ such that $\phi(i, j) \le O(i \cdot j^2)$. 
    \label{lem:injection}
\end{lemma}
\begin{proof}
    Take $i$'s binary representation $a_1 \circ a_2 \circ \cdots \circ a_{p}$ and $j$'s binary representation $b_1 \circ b_2 \circ \cdots \circ b_q$ (here, $a_1$ and $b_1$ are the most significant bits of $i$ and $j$, respectively), and construct $\phi(i, j)$ to have binary representation 
    $$1\circ b_1\circ 1\circ b_2\circ 1\circ b_3\circ \cdots\circ 1\circ b_1\circ 0\circ a_1\circ a_2\circ \ldots\circ a_p,$$
    where again the digits read from most significant bit to least significant bit. The map $\phi(i, j)$ is an injection by design since one can straightforwardly recover $i$ and $j$ from $\phi(i, j)$'s binary representation. On the other hand, 
    $$\log_2 \phi(i, j) \le \log_2 i + 2 \log_2 j + O(1),$$
    which means that $\phi(i, j) \le O(i \cdot j^2)$, as desired.
\end{proof}

\paragraph{The algorithm.} We now describe our insertion algorithm. Break the array $A$ of size $n$ into disjoint arrays $A_1, A_2, \ldots, A_{\lceil \log n \rceil}$ satisfying $|A_{i + 1}| = |A_i|/2 \pm 1$.\footnote{The $\pm 1$'s are needed so that it is possible to satisfy $|A_1| + |A_2| + \cdots + A_{\lceil \log n \rceil} = n$. Note that, in this context, $a \pm 1$ means one of $a - 1, a, a + 1$.} 

We will simulate a two-dimensional probe sequence $\{h_{i, j}\}$, where probe $h_{i, j}(x)$ is a random slot in array $A_i$. In particular, we map the entries of the two-dimensional sequence $\{h_{i, j}\}$ to those of a one-dimensional sequence $\{h_i\}$ by defining 
$$h_{\phi(i, j)}(x) := h_{i, j}(x),$$
where $\phi$ is the map from Lemma \ref{lem:injection}. This means that the probe complexity of an element $x$ placed in slot $h_{i, j}(x)$ is $O(i \cdot j^2)$.

We break the  $n - \lfloor \delta n \rfloor$ insertions into \defn{batches} $\mathcal{B}_0, \mathcal{B}_1, \mathcal{B}_2, \ldots$. Batch $\mathcal{B}_0$ fills array $A_1$ to have $\lceil 0.75 |A_1| \rceil$ elements, where each element $x$ is inserted using the first available slot in the probe sequence $h_{1, 1}(x), h_{1, 2}(x), h_{1, 3}(x), \ldots$. For $i \ge 1$, batch $\mathcal{B}_i$ consists of 
\begin{equation} |A_i| - \lfloor \delta |A_i| /2  \rfloor - \lceil 0.75 \cdot |A_i|  \rceil + \lceil 0.75 \cdot |A_{i + 1}|\rceil
\label{eq:batchsize}
\end{equation} 
insertions, all of which are placed in arrays $A_i$ and $A_{i + 1}$. (The final batch may not finish, since we run out of insertions.) For $i \ge 0$, the guarantee at the end of the batch  $\mathcal{B}_i$ is that each $A_j$ satisfying $j \in \{1, \ldots, i\}$ contains exactly $|A_j| - \lfloor \delta |A_j| / 2 \rfloor$ elements, and that $A_{i + 1}$ contains exactly $\lceil 0.75 \cdot |A_{i + 1}|\rceil$ elements. Note that this guarantee forces the batch size to be given by \eqref{eq:batchsize}. Additionally, because the total number of insertions is $n - \lfloor \delta n\rfloor$, and because each batch $\mathcal{B}_i$ leaves at most $O(n / 2^i) + \delta n / 2$ remaining free slots in the full array $A$, the insertion sequence is guaranteed to finish within $O(\log \delta^{-1})$ batches.

Let $c$ be a parameter that we will later set to be a large positive constant, and define the function 
$$f(\epsilon) = c \cdot \min(\log^2 \epsilon^{-1}, \log \delta^{-1}).$$
We now describe how to implement the insertion of an element $x$ during a batch $\mathcal{B}_i$, $i \ge 1$. Suppose that, when the insertion occurs, $A_i$ is $1 - \epsilon_1$ full and $A_{i + 1}$ is $1 - \epsilon_2$ full. There are three cases:
\begin{enumerate}
    \item If $\epsilon_1 > \delta/2$ and $\epsilon_2 > 0.25$, then $x$ can go in either of $A_i$ or $A_{i + 1}$ and is placed as follows: if any of the positions 
$$h_{i, 1}(x), h_{i, 2}(x), \ldots, h_{i, f(\epsilon_1)}(x)$$  
are free in $A_i$, then $x$ is placed in the first such free slot; and, otherwise, $x$ is placed in the first free slot from the sequence of positions
$$h_{i + 1, 1}(x), h_{i + 1, 2}(x), h_{i + 1, 3}(x), \ldots.$$
    \item If $\epsilon_1 \le \delta/2$, then $x$ must be placed in $A_{i + 1}$, and $x$ is 
    placed in the first free slot from the sequence of positions
$$h_{i + 1, 1}(x), h_{i + 1, 2}(x), h_{i + 1, 3}(x), \ldots.$$
    \item Finally, if $\epsilon_2 \le 0.25$,  then $x$ must be placed in $A_{i}$, and $x$ is 
    placed in the first free slot from the sequence of positions
$$h_{i, 1}(x), h_{i, 2}(x), h_{i, 3}(x), \ldots.$$
\end{enumerate}
We refer to the final case as the \defn{expensive case} since $x$ is inserted into a potentially very full array $A_i$ using uniform probing. We shall see later, however, that this case is very rare: with probability $1 - O(1 / |A_i|^2)$, the case never occurs during batch $\mathcal{B}_i$.

Note that Cases 2 and 3 are disjoint (only one of the two cases can ever occur in a given batch) by virtue of the fact that, once $\epsilon_1 \le \delta/2$ and $\epsilon_2 \le 0.25$ hold simultaneously, then the batch is over.

\paragraph{Bypassing the coupon-collector bottleneck.} Before we dive into the analysis, it is helpful to understand at a very high level how our algorithm is able to bypass the ``coupon-collector'' bottleneck faced by uniform probing. In uniform probing, each probe can be viewed as sampling a random coupon (i.e., slot); and standard coupon-collector lower bounds say that at least $\Omega(n \log \delta^{-1})$ probes need to be made if a $(1 - \delta)$-fraction of coupons are to be collected. This prohibits uniform probing (or anything like uniform probing) from achieving an amortized expected probe complexity better than $O(\log \delta^{-1})$. 

A critical feature of our algorithm is the way in which it decouples each key's \emph{insertion probe complexity} (i.e., the number of probes made while inserting the key) from its \emph{search probe complexity} (i.e., the number of probes needed to find the key). The latter quantity, of course, is what we typically refer to simply as \emph{probe complexity}, but to avoid ambiguity in this section, we will sometimes call it \emph{search probe complexity}.

The insertion algorithm will often probe much further down the probe sequence than the position it ends up using. It might seem unintutive at first glance that such insertion probes could be useful, but as we shall see, they are the key to avoiding the coupon-collecting bottleneck---the result is that most coupons contribute \emph{only} to the insertion probe complexity, and not to the search probe complexity. 

To see the decoupling in action, consider an insertion in batch $\mathcal{B}_1$ in which $A_1$ is, say, a $(1 - 2 \delta^{-1})$-fraction full, and $A_2$ is, say, a $0.6$-fraction full. The insertion makes $\Theta(f(\delta^{-1})) = \Theta(\log \delta^{-1})$ probes to $A_1$, but most likely they all fail (each probe has only an $O(\delta)$ probability of success). The insertion then looks in $A_2$ for a free slot, and most likely ends up using a position of the form $h_{\phi(2, j)}$ for some $j = O(1)$, resulting in search probe complexity $O(\phi(2, j)) = O(1)$. So, in this example, even though the insertion probe complexity is $\Theta(\log \delta^{-1})$, the search probe complexity is $O(1)$.

The coupon-collector bottleneck is also what makes \emph{worst-case} expected insertion (and search) bounds difficult to achieve. We know that $\Theta(n \log \delta)$ total coupons must be collected, and intuitively it is the final insertions (i.e., those that take place a high load factors) that must do most of the collecting. After all, how can insertions at low load factors make productive use of more than a few coupons? This is what dooms algorithms such as uniform probing to have a worst-case expected insertion time of $O(\delta^{-1})$.

Our algorithm also circumvents this bottleneck: even though $\Theta(n \log \delta^{-1})$ total coupons are collected, no insertion has an expected contribution of more than $O(\log \delta^{-1})$. This means that even insertions that take place at low load factors need to be capable of `productively' making use of $\Theta(\log \delta^{-1})$ probes/coupons. How is this possible? The key is that a constant-fraction of insertions $x$ have the following experience: when $x$ is inserted (in some batch $\mathcal{B}_i$), the array $A_i$ is \emph{already} almost full (so $x$ can productively sample $O(\log \delta^{-1})$ probes/coupons in that array), but the next array $A_{i + 1}$ is not very full (so $x$ can go there in the likely event that none of the $O(\log \delta^{-1})$ probes/coupons in $A_i$ pay off).  This is how the algorithm is able to spread the coupon collecting (almost evenly!) across $\Theta(n)$ operations.


\paragraph{Algorithm Analysis. } We begin by analyzing the probability of a given batch containing insertions in the expensive case.
\begin{lemma}
With probability at least $1 - O(1 / |A_i|^2)$, none of the insertions in batch $\mathcal{B}_i$ are in the expensive case (i.e., Case 3).
\label{lem:failed}
\end{lemma}
\begin{proof}
    Let $m$ denote the size of array $A_i$. We may assume that $m = \omega(1)$, since otherwise the lemma is trivial. For $j \in \{2, 3, \ldots, \lceil \log \delta^{-1} \rceil\}$, let $T_j$ be the time window during batch $\mathcal{B}_i$ in which $A_i$ goes from having $ \lfloor m /2^j \rfloor$ free slots to $\max(\lfloor m / 2^{j + 1} \rfloor, \lfloor \delta m / 2 \rfloor)$ free slots.

    Each insertion during $T_j$ is guaranteed to be in one of Cases 1 or 3, so the insertion makes at least $f(2^{-j})$ probe attempts in $A_i$, each of which has at least $2^{-(j + 1)}$ probability of succeeding. If $f(2^{-j}) > 100 \cdot  2^j$, then each insertion during time window $T_j$ has probability at least $1 - (1 - 1/2^{j + 1})^{100 \cdot 2^j} > 0.99$ of being placed in array $A_i$. Otherwise, if $f(2^{-j}) < 100\cdot 2^j$, then the insertion has a $\Theta(f(2^{-j}) / 2^j)$ probability of being placed in array $A_i$. Thus, in general, each insertion in $T_j$ uses array $A_i$ with probability at least
    \begin{equation}\min(0.99, \Theta(f(2^{-j}) / 2^j)).
    \label{eq:useAi}
    \end{equation}
    It follows that
    $$\E[|T_j|] \le \frac{m / 2^{j + 1}}{\min(0.99, \Theta(f(2^{-j}) / 2^j))} + O(1) \le \Theta\left( \frac{m}{f(2^{-j})} \right) + 1.02 \cdot m / 2^{j + 1} + O(1).$$
    
    Since \eqref{eq:useAi} holds for each insertion in $T_i$ independently of how the previous insertions behave, we can apply a Chernoff bound to conclude that, with probability at least $1 - 1 / m^3$,
    \begin{equation}
        |T_j| \le \E[|T_j|] + O(\sqrt{m \log m}) \le O\left( \frac{m}{f(2^{-j})} \right) + 1.02 \cdot m / 2^{j + 1} + O(\sqrt{m \log m}).
        \label{eq:Ti}
    \end{equation}
    With probability at least $1 - O(1 / m^2)$, \eqref{eq:Ti} holds for every window $T_j$. 

    Thus, treating $c$ as a parameter (rather than a constant), we have that
    \begin{align*}
    \sum_j |T_j| & \le \sum_{j = 2}^{\lceil \log \delta^{-1} \rceil} \left(O\left( \frac{m}{f(2^{-j})}\right) + 1.02 \cdot m / 2^{j + 1} + O(\sqrt{m \log m})\right) \\
    & \le 0.26 \cdot m + m \cdot O\left( \sum_{j = 2}^{\lceil \log \delta^{-1} \rceil} \frac{1}{f(2^{-j})}\right) + O(1) \\
    & = 0.26 \cdot m + m \cdot O\left( \sum_{j = 2}^{\lceil \log \delta^{-1} \rceil} \frac{1}{c\min(j^2, \log \delta^{-1})}\right) + O(1) \\
    & \le 0.26 \cdot m + \frac{m}{c} \cdot O\left( \sum_{j = 2}^{\infty} \frac{1}{j^2} + \sum_{j = 2}^{\lceil \log \delta^{-1} \rceil}  \frac{1}{\log \delta^{-1}}\right) + O(1)\\
    & \le 0.26 \cdot m + \frac{m}{c} \cdot O(1) + O(1). \\
    \end{align*}
    
    If we set $c$ to be a sufficiently large positive constant, then it follows that $\sum_j |T_j| < 0.27 \cdot m + O(1)$. However, the first $0.27 \cdot m$ insertions during batch $\mathcal{B}_i$ can fill array $A_{i + 2}$ to at most a $0.54 + o(1) < 0.75$ fraction full (recall, in particular, that we have $m = \omega(1)$ without loss of generality). This means that none of the insertions during time windows $T_1, T_2, \ldots$ are in Case 3 (the expensive case). On the other hand, after time windows $T_1, T_2, \ldots$ are complete, the remaining insertions in the batch are all in Case 2. Thus, with probability at least $1 - O(1 / m^2)$, none of the insertions in the batch are in Case 3.
\end{proof}

Next, we bound the expected search probe complexity for a given insertion within a given batch.

\begin{lemma}
The expected search probe complexity for an insertion in batch $\mathcal{B}_i$ is $O(1 + i)$.
\label{lem:probe}
\end{lemma}
\begin{proof}
Insertions in batch $0$ have search probe complexities of the form $\phi(1, j) = O(j^2)$ where $j$ is a geometric random variable with mean $O(1)$. They therefore have expected probe complexities of $O(1)$. For the rest of the proof, let us assume that $i > 0$.

Let $x$ be the element being inserted. Let $C_j$, $j \in \{1, 2, 3\}$, be the indicator random variable for the event that $x$'s insertion is in Case $j$. Let $D_k$, $k \in \{1, 2\}$ be the indicator random variable for the event that $x$ ends up in array $A_{i + k - 1}$. Finally, let $Q$ be the search probe complexity of $x$. We can break $\E[Q]$ into
\begin{align}
    \E[Q] & = \E[Q C_1 D_1] + \E[Q C_1 D_2] + \E[Q C_2] + \E[Q C_3] \\
    & \le \E[Q  C_1 D_1] + \E[Q D_2] + \E[Q C_3], \label{eq:start}
\end{align}
where the final inequality uses the fact that $C_2$ implies $D_2$ and thus that $\E[Q C_1 D_2] + \E[Q C_2] \le \E[Q D_2]$.

To bound $\E[Q  C_1 D_1]$, observe that
\begin{align}
    \E[Q  C_1 D_1] & \le \E[Q  D_1 \mid C_1] \\
    & = \E[Q \mid D_1, C_1] \cdot \Pr[D_1 \mid C_1] \label{eq:intermed} 
\end{align}
Suppose that, when $x$ is inserted, array $A_i$ is $1 - \epsilon$ full and that $x$ uses Case 1. Then the only positions that $x$ considers in $A_i$ are $h_{i, 1}, \ldots, h_{i, f(\epsilon^{-1})}$. The probability that any of these positions are free is at most $O(f(\epsilon^{-1}) \cdot \epsilon)$. And, if one is free, then the resulting search probe complexity $Q$ will be at most $\phi(i, f(\epsilon^{-1}) \le O(i f(\epsilon^{-1})^2)$. Thus \eqref{eq:intermed} satisfies
\begin{align*}
    & \E[Q \mid D_1, C_1] \cdot \Pr[D_1 \mid C_1] \\
    & \le O(i f(\epsilon^{-1})^2) \cdot O( f(\epsilon^{-1}) \epsilon) \\
    & \le O(i \epsilon f(\epsilon^{-1})^3)\\
    & \le O(i \epsilon \log^6 \epsilon^{-1}) \\
    & \le O(i).
\end{align*}

To bound $\E[Q D_2]$, recall that $D_2$ can only occur if $A_{i + 1}$ is at most a $0.75$ fraction full. Thus, if $D_2$ occurs, then $x$ will have search probe complexity $\phi(i + 1, j)$ where $j$ is at most a geometric random variable with mean $O(1)$. We can therefore bound $\E[Q D_2]$ by
$$\E[Q D_2] \le \E[Q \mid D_2] \le \E[\phi(i + 1, j)],$$
where $j$ is a geometric random variable with mean $O(1)$. This, in turn, is at most
$$O(\E[i \cdot j^2]) = O(i).$$

Finally, to bound $\E[Q C_3]$, observe that
\begin{equation}
   \E[Q C_3]  = \E[Q \mid C_3] \cdot \Pr[C_3].
   \label{eq:QC3}
\end{equation}
By Lemma \ref{lem:failed}, we have $\Pr[C_3] = O(1 / |A_i|^2)$. Since Case 3 inserts $x$ into $A_i$ using the probe sequence $h_{\phi(i, 1)}, h_{\phi(i, 2)}, \ldots$, the search probe complexity of $x$ will end up being given by $\phi(i, j) = O(i \cdot j^2)$ where $j$ is a geometric random variable with mean $O(|A_i|)$. This implies a bound on $\E[Q \mid C_3]$ of $O(i \cdot |A_i|^2)$.  Thus, we can bound \eqref{eq:QC3} by
$$\E[Q \mid C_3] \cdot \Pr[C_3] \le O(i \cdot |A_i|^2 / |A_i|^2) = O(i).$$

Having bounded each of the terms in \eqref{eq:start} by $O(i)$, we can conclude that $\E[Q] = O(i)$, as desired.
\end{proof}

Finally, we bound the worst-case expected insertion time by $O(\log \delta^{-1})$.
\begin{lemma}
The worst-case expected time for an insertion is $O(\log \delta^{-1})$.
\label{lem:worst}
\end{lemma}
\begin{proof}
Insertions in batch $0$ take expected time $O(1)$, since they make $O(1)$ expected probes into $A_1$. Now consider an insertion in some batch $\mathcal{B}_i$, $i \ge 1$. If the insertion is in either of Cases 1 or 2, then the insertion makes at most $f(\delta^{-1})$ probes in $A_i$ and at most $O(1)$ expected probes in $A_{i + 1}$. The expected insertion time in each of these cases is therefore at most $f(\delta^{-1}) = O(\log \delta^{-1})$. Finally, each insertion has probability at most $1 / |A_i|^2$ of being in Case 3 (by Lemma \ref{lem:failed}), and the expected insertion time in Case 3 can be bounded by $O(|A_i|)$ (since we probe repeatedly in $A_i$ to find a free slot). Therefore, the contribution of Case 3 to the expected insertion time is at most $O(1 / |A_i|) = O(1)$.
\end{proof}

Putting the pieces together, we prove Theorem \ref{thm:main-intro}
\begin{proof}
By Lemmas \ref{lem:probe}, the insertions in $\mathcal{B}_i$ each have expected search probe complexity $O(i)$. Since there are $O(\log \delta^{-1})$ batches, this implies a worst-case expected search probe complexity of $O(\log \delta^{-1})$. And, since the $|\mathcal{B}_i|$s are geometrically decreasing, the amortized expected search probe complexity overall is $O(1)$. Finally, by Lemma \ref{lem:worst}, the worst-case expected time per insertion is $O(\log \delta^{-1})$. This completes the proof of the theorem.
\end{proof}

\section{Funnel Hashing}\label{sec:funnel}


In this section, we construct a \emph{greedy} open-addressing scheme that achieves $O(\log^2 \delta)$ worst-case expected probe complexity, and high-probability worst-case probe complexity $O(\log^2 \delta + \log \log n)$. As we shall see, the high-probability worst-case bound is optimal.

\thmgreedy*
\begin{proof}

Throughout the section, we assume without loss of generality that $\delta \leq 1/8.$ Let $\alpha = \left\lceil 4\log \delta^{-1} + 10\right \rceil$ and $\beta = \left\lceil 2\log \delta^{-1}\right \rceil$. 

The greedy open-addressing strategy that we will use in this section is as follows. First, we split array $A$ into two arrays, $A'$ and a \defn{special array} denoted $A_{\alpha + 1},$ where $\left\lfloor 3\delta n /4\right \rfloor \geq |A_{\alpha+1}| \geq \left\lceil \delta n /2\right \rceil,$ with the exact size chosen so that $|A'|$ is divisible by $\beta$. Then, split $A'$ into $\alpha$ arrays $A_1, \ldots, A_{\alpha}$ such that $|A_i| = \beta a_i$, satisfying $a_{i+1} = 3a_i/4 \pm 1.$ That is, the size of each array is a multiple of $\beta$ and they are (roughly) geometrically decreasing in size. Note that, for $i \in [\alpha - 10]$,
$$\sum_{j > i} |A_j|  \ge ((3/4) + (3/4)^2 + \cdots + (3/4)^{10}) \cdot |A_i| > 2.5 |A_i|.$$ 


Each array $A_i$ with $i \in [\alpha]$ is further subdivided into arrays $A_{i, j}$, each of size $\beta.$ We define an \defn{attempted insertion} of a key $k$ into $A_i$ (for $i \in [\alpha]$) as follows:
\begin{enumerate}
    \item Hash $k$ to obtain a subarray index $j \in \left[\frac{|A_i|}{\beta}\right].$
    \item Check each slot in $A_{i, j}$ to see if any are empty.
    \item If there is an empty slot, insert into the first one seen, and return success. Otherwise, return fail.
\end{enumerate}
To insert a key $k$ into the overall data structure, we perform attempted insertions on each of $A_1, A_2, \ldots, A_{\alpha}$, one after another, stopping upon a successful attempt. Each of these $\alpha = O(\log \delta^{-1})$ attempts probes up to $\beta = O(\log \delta^{-1})$ slots. If none of the attempts succeed, then we insert $k$ into the special array $A_{\alpha + 1}$. The special array $A_{\alpha + 1}$ will follow a different procedure than the one described above---assuming it is at a load factor of at most $1/4$, it will ensure $O(1)$ expected probe complexity and $O(\log \log n)$ worst-case probe complexity. Before we present the implementation of $A_{\alpha + 1}$, we will first analyze the behaviors of $A_1, A_2, \ldots, A_{\alpha}$.

At a high level, we want to show that each $A_i$ fills up to be almost full over the course of the insertions. Critically, $A_i$ does not need to give any guarantees on the probability of any specific insertion succeeding. All we want is that, after the insertions are complete, $A_{i}$ has fewer than, say, $\delta|A_i|/64$ free slots. 

\begin{lemma}
    For a given $i \in [\alpha]$, we have with probability $1 - n^{-\omega(1)}$ that, after $2|A_i|$ insertion attempts have been made in $A_i$, fewer than $\delta|A_i|/64$ slots in $A_i$ remain unfilled.
    \label{lem:overflow}
\end{lemma}

\begin{proof}
    Since each insertion attempt selects a uniformly random $A_{i, j}$ to use, out of $|A_i| / \beta$ options, the expected number of times that a given $A_{i, j}$ is used is $2\beta.$ Letting the number of attempts made to insert into $A_{i, j}$ be $X_{i, j}$, we have by a Chernoff bound that
    $$\Pr[X_{i, j} < \beta] = \Pr[X_{i, j} < (1-1/2) \mathbb{E}[X_{i, j}]] = e^{-2^2 \beta/2} \leq e^{-4 \log \delta^{-1}} = \delta^{4} \leq \frac{1}{128} \delta.$$

    Note that, since we always insert into the subarray we choose if it has empty slots, the only scenario in which $A_{i, j}$ remains unfilled is if $X_{i, j} < \beta$. Therefore, the expected number of subarrays that remain unfilled is at most $\frac{1}{128} \delta \left(\frac{|A_i|}{\beta}\right),$ and, consequently, the expected number of subarrays that become full is $\left(1-\frac{1}{128} \delta\right) \left(\frac{|A_i|}{\beta}\right)$.
    
    Define $Y_{i, k}$ to be the random number in $[|A_i|/\beta]$ such that the $k$th insertion attempt into $A_i$ uses subarray $A_{i, Y_{i, k}}$. Let $f(Y_{i, 1}, \ldots, Y_{i, 2|A_i|})$ denote how many subarrays $A_{i, j}$ remain unfilled after $2|A_i|$ insertion attempts have been made. Changing the outcome of a single $Y_{i, k}$ changes $f$ by at most $2$---one subarray may become unfilled and one may become filled. Also, by the above, $\mathbb{E}[f(Y_{i, 1}, \ldots, Y_{i, 2|A_i|})] = \frac{1}{128} \delta \left(\frac{|A_i|}{\beta}\right)$. Therefore, by McDiarmid's inequality,
    \begin{align*}
        \Pr\left[f(Y_{i, 1}, \ldots, Y_{i, 2|A_i|}) \geq\frac{1}{64}\delta \left(\frac{|A_i|}{\beta}\right)\right] \leq \exp\left(-\frac{2\left(\frac{1}{128}\delta\frac{|A_i|}{\beta}\right)^2}{2|A_i|}\right) = \exp\left(-|A_i| O(\beta^2\delta^2)\right).
    \end{align*}

    Since $|A_i| = n \poly(\delta)$ and $\delta = n^{o(1)},$ we have that $$|A_i| O(\beta^2\delta^2) = n^{1 - o(1)},$$
    so the probability that more than a $\frac{\delta}{64}$-fraction of the subarrays in $A_i$ remain unfilled is $\exp(-n^{1 - o(1)}) = 1 / n^{-\omega(1)}$. All subarrays are the same size, so, even if these unfilled subarrays remain completely empty, we still have that $\frac{1}{64} \delta$ of the total slots remain unfilled, as desired.

\end{proof}




As a corollary, we can obtain the following statement about $A_{\alpha + 1}$: 

\begin{lemma}
    With probability $1 - n^{-\omega(1)}$, the number of keys inserted into $A_{\alpha+1}$ is fewer than $\frac{\delta}{8}n$.
\end{lemma}

\begin{proof}
   
Call $A_i$ \emph{fully explored} if at least $2|A_i|$ insertion attempts are made to $A_i$. By Lemma \ref{lem:overflow}, we have with probability $1 - n^{-\omega(1)}$ that every fully-explored $A_i$ is at least $(1 - \delta/64)$ full. We will condition on this property for the rest of the lemma.

Let $\lambda \in [\alpha]$ be the largest index such that $A_\lambda$ receives fewer than $2|A_\lambda|$ insertion attempts (or $\lambda = \texttt{null}$ if no such index exists). We will handle three cases for $\lambda.$

First, suppose that $\lambda \leq \alpha-10$. By definition, we know that, for all $\lambda < i \in [\alpha]$, $A_i$ is fully explored, and therefore that $A_i$ contains at least $|A_i|(1-\delta/64)$ keys. The total number of keys in $A_i$, $i > \lambda$, is therefore at least
    $$(1-\delta/64) \sum_{i=\lambda+1}^{\alpha} |A_i| \ge 2.5(1-\delta/64)|A_{\lambda}|, $$
    contradicting the fact that at most $2|A_{\lambda}|$ insertions are made in total for all arrays $A_i$ with $i \geq \lambda$ (recall by the construction of our algorithm that we must first try [and fail] to insert into \(A_{\lambda}\) before inserting into \(A_i\) for any \(i \geq \lambda\)). This case is thus impossible, and we are done.

    Next, suppose that $\alpha-10 < \lambda \leq \alpha$. In this case, fewer than $2|A_{\alpha-10}| < n\delta/8$ keys are attempted to be inserted into any $A_i$ with $i \geq \lambda,$ including $i = \alpha+1$, and we are done.

    Finally, suppose that $\lambda = \texttt{null}.$ In this case, each $A_i$, $i \in [\alpha]$, has at most $\delta|A_i|/64$ empty slots. Therefore, the total number of empty slots at the end of all insertions is at most
    $$|A_{\alpha+1}| + \sum_{i=1}^{\alpha} \frac{\delta|A_i|}{64} = |A_{\alpha+1}| + \frac{\delta|A'|}{64} \leq \frac{3n\delta}{4} + \frac{n\delta}{64} < n\delta,$$
    which contradicts the fact that, after $n(1-\delta)$ insertions, there are at least $n \delta$ slots empty. This concludes the proof.
\end{proof}

Now, the only part left is to implement the $\le \delta n / 8$ insertions that reach $A_{\alpha+1}$. We must do so with $O(1)$ expected probe complexity and with $O(\log \log n)$ worst-case probe complexity, while incurring at most a $1 / \poly(n)$ probability of hash-table failure.

We implement $A_{\alpha+1}$ in two parts. That is, split $A_{\alpha+1}$ into two subarrays, $B$ and $C$, of equal ($\pm 1$) size. To insert, we first try to insert into $B$, and, upon failure, we insert into $C$ (an insertion into $C$ is guaranteed to succeed with high probability). $B$ is implemented as a uniform probing table, and we give up searching through $B$ after $\log \log n$ attempts. $C$ is implemented as a two-choice table with buckets of size $2 \log \log n$.

Since $B$ has size $|A_{\alpha + 1}|/2 \ge \delta n / 4$, its load factor never exceeds $1/2$. Each insertion into $A_{\alpha + 1}$ makes $\log \log n$ random probes in $B$, each of which has at least a $1/2$ probability of succeeding. The expected number of probes that a given insertion makes in $B$ is therefore $O(1)$, and the probability that a given insertion tries to use $B$ but fails (therefore moving on to $C$) is at most $1 / 2^{\log \log n} \le 1 / \log n$.

On the other hand, $C$ is implemented as a two choice table with buckets of size $2\log \log n$. Each insertion hashes to two buckets $a$ and $b$ uniformly at random, and uses a probe sequence in which it tries the first slot of $a$, the first slot of $b$, the second slot of $a$, the second slot of $b$, and so on. The effect of this is that the insertion ends up using the emptier of the two buckets (with ties broken towards $a$). If both buckets are full, our table fails. However, with high probability, this does not happen, by the following classical power-of-two-choices result \cite{berenbrink2000tchashing}:

\begin{theorem}
    If $m$ balls are placed into $n$ bins by choosing two bins uniformly at random for each ball and placing the ball into the emptier of the two bins, then the maximum load of any bin is $m/n + \log \log n + O(1)$ with high probability in $n$.
\end{theorem}

Applying this theorem to our setting, we can conclude that, with high probability in $|A_{\alpha + 1}| / \log \log n$, and therefore with high probability in $n$, no bucket in $C$ ever overflows. This ensures the correctness of our implementation of $A_{\alpha + 1}$.

Since each insertion in $A_{\alpha + 1}$ uses $C$ with probability at most $1/\log n$, and there are at most $2 \log \log n$ slots checked in $C$, the expected time that each insertion spends in $C$ is at most $o(1)$. Thus, insertions that reach $A_{\alpha+1}$ take expected time $O(1)$ and worst-case time $O(\log \log n)$.



Since we only attempt to insert into $\beta$ slots for each $A_i$ (a single bucket), the probe complexity of a given insertion is at most $\beta \alpha + f(A_{\alpha+1})= O(\log^2 \delta^{-1} + f(A_{\alpha+1}))$, where $f(A_{\alpha+1})$ is the number of probes made in $A_{\alpha+1}$. This implies a worst-case expected \probecomplexity of $O(\log^2 \delta^{-1})$ and a high-probability worst-case probe complexity of $O(\log^2 \delta^{-1} + \log \log n)$.

We now only have left to prove the amortized expected \probecomplexity. The expected number of probes we make into each subarray is at most $c \log \delta^{-1}$ for some constant $c$ (including for $A_{\alpha+1}$), and we first insert into $A_1$, then $A_2$, and so on. Thus, the total expected \probecomplexity across all keys is at most
$$|A_1| \cdot c \log \delta^{-1} + |A_2| \cdot 2c \log \delta^{-1} + |A_3| \cdot 3c \log \delta^{-1} + \cdots + |A_{\alpha + 1}| \cdot (\alpha+1)\log \delta^{-1}.$$
Since the $A_i$'s are geometrically decreasing in size with the exception of $A_{\alpha+1}$, which itself is only $O(n \delta)$ in size, the above sum is dominated (up to a constant factor) by its first term. The total expected \probecomplexity across all keys is thus $O(|A_1| \log \delta^{-1}) = O(n \log \delta^{-1}),$ implying that the amortized expected probe complexity is $O(n \log \delta^{-1}/(n(1-\delta))) = O(\log \delta^{-1}),$ as desired.
This completes the proof of Theorem \ref{thm:greedy-intro}.
\end{proof}


\section{A Lower Bound for Greedy Algorithms}\label{lower:greedy}

In this section we prove that the $O(\log^2 \delta^{-1})$ expected-cost bound from Theorem \ref{thm:greedy-intro} is optimal across all greedy open-addressed hash tables.

\begin{theorem}
Let $n \in \mathbb{N}$ and $\delta \in (0, 1)$ be parameters, where $\delta$ is an inverse power of two. Consider a greedy open-addressed hash table with capacity $n$. If $(1 - \delta) n$ elements are inserted into the hash table, then the final insertion must take expected time $\Omega(\log^2 \delta^{-1})$.
\label{thm:lowergreedy}
\end{theorem}

Our proof of Theorem \ref{thm:lowergreedy} will make use of Yao's lower bound on \emph{amortized} insertion time \cite{yao1985uniform}:

\begin{proposition}[Yao's Theorem \cite{yao1985uniform}]
    Let $n \in \mathbb{N}$ and $\delta \in (0, 1)$ be parameters. Consider a greedy open-addressed hash table with capacity $n$. If $(1 - \delta) n$ elements are inserted into the hash table, then the amortized expected time per insertion must be $\Omega(\log \delta^{-1})$.
    \label{prop:yao}
\end{proposition}

Building on Proposition \ref{prop:yao}, we can obtain the following critical lemma:
\begin{lemma}
    There exists a universal positive constant $c > 0$ such that the following is true: For any values of $\delta,n$, there exists some integer $1 \le i \le \log \delta^{-1}$ such that the $(1 - 1/2^i)n$-th insertion has expected cost at least $c i \log \delta^{-1}$.
    \label{lem:prefixlower}
\end{lemma}
\begin{proof}
    Let $c$ be a sufficiently small positive constant, and suppose for contradiction that the lemma does not hold. Let $q_j$ be the expected cost of the $j$-th insertion. Because the hash table uses greedy open addressing, we know that the $q_j$s are monotonically increasing. It follows that
    $$\E[\sum q_j] \ge \sum_{i \in [1, \log \delta^{-1}]} \frac{n}{2^i}  \cdot q_{(1 - 1/2^i)n},$$
    which by assumption is at most
     $$\sum_{i \in [1, \log \delta^{-1}]} \frac{n}{2^i} \cdot c \cdot i \log \delta^{-1} \le c n \cdot O(\log \delta^{-1}).$$
    Setting $c$ to be a sufficiently small positive constant contradicts Proposition \ref{prop:yao}.
\end{proof}

\begin{lemma}
    Let $c$ be the positive constant from Lemma \ref{lem:prefixlower}. Then, the final insertion takes expected time at least 
    $$\sum_{j = 1}^{\log \delta^{-1}} cj.$$
    \label{lem:sumlower}
\end{lemma}
\begin{proof}
By Lemma \ref{lem:prefixlower}, there exists an integer $1 \le i \le \log \delta^{-1}$ such that the $(1 - 1/2^i)n$-th insertion has expected cost at least $c i \log \delta^{-1}$. If $i = \log \delta^{-1}$, then we are done. Otherwise, we can complete the proof by strong induction on $n$, as follows.

Let $S$ denote the set of occupied positions after the $(1 - 1/2^i)n$-th insertion. Condition on some outcome for $S$, and define the \defn{second-layer cost} for future insertions to be the expected number of probes that the insertion makes to slots $[n] \setminus S$. To analyze the second-layer costs, we can imagine that the slots $[n] \setminus S$ are the only slots in the hash table, and that the slots in $S$ are removed from each element's probe sequence. This new ``compressed'' hash table has size $n / 2^i$ and is set to receive $n/2^i - \delta n = (n/2^i) \cdot (1 - \delta 2^i)$ insertions that are each implemented with greedy open addressing. It follows by induction that the final insertion in the ``compressed'' hash table has expected cost at least 
\begin{equation}
    \sum_{j = 1}^{\log \delta^{-1} - i} cj.
    \label{eq:secondlayer}
\end{equation}
This is equivalent to saying that the final insertion in the full hash table has expected \emph{second-layer cost} at least \eqref{eq:secondlayer}. Moreover, although \eqref{eq:secondlayer} was established conditioned on some specific outcome for $S$, since it holds for each individual outcome, it also holds without any conditioning at all. 

Finally, in addition to the second-layer cost, the final insertion must also perform at least $c i \log n$ expected probes even just to find any slots that are not in $S$. (This is due to our earlier application of Lemma \ref{lem:prefixlower}.) It follows that the total expected cost incurred by the final insertion is at least
\begin{equation*}
   c i \log \delta^{-1} + \sum_{j = 1}^{\log \delta^{-1} - i} cj \ge  \sum_{j = 1}^{\log \delta^{-1}} cj,
\end{equation*}
as desired.
\end{proof}

Finally, we can prove Theorem \ref{thm:lowergreedy} as a corollary of Lemma \ref{lem:sumlower}.

\begin{proof}[Proof of Theorem \ref{thm:lowergreedy}]
    By Lemma \ref{lem:sumlower}, there exists a positive constant $c$ such that the expected cost incurred by the final insertion is at least 
    $$\sum_{j = 1}^{\log \delta^{-1}} cj = \Omega(\log^2 \delta^{-1}).$$
\end{proof}

\section{Lower Bounds For Open Addressing Without Reordering}\label{sec:wclb}

In this section, we give two lower bounds that apply not just to greedy open-addressed hash tables but to any open-addressed hash table that does not perform reordering. Our first result is a lower bound of $\Omega(\log \delta^{-1})$ on worst-case expected \probecomplexity (matching the upper bound from Theorem \ref{thm:main-intro}). Our second result is a lower bound of $\Omega(\log^2 \delta^{-1} + \log \log n)$ on (high-probability) worst-case probe complexity (matching the upper bound from Theorem \ref{thm:greedy-intro}, which is achieved by a greedy scheme). 



For the following proofs, we assume that the probe sequences for keys are iid random variables. This is equivalent to assuming that the universe size is a large polynomial and then sampling the keys at random (with replacement); with high probability, such a sampling procedure will not sample any key twice.



\subsection{Common Definitions}
Both lower bounds will make use of a shared set of definitions:
\begin{itemize}
    \item Let $m = n(1-\delta)$.
    \item Let $k_1, k_2, \ldots, k_m$ be the set of keys to be inserted.     \item Let $H_i(k_j)$ be $i$-th entry in the probe sequence for $k_j$. Because the distribution of $H_i(k_j)$ is the same for all $k_j$, we will sometimes use $H_i$ as a shorthand. We will also use $h_i$ to refer to a (non-random) specific outcome for $H_i$.
    \item Let $\mathcal{H}_c(k_j) = \{H_i(k_j): i \in [c]\}$ denote the set consisting of the first $c$ probes made by $k_j$. Again, because $\mathcal{H}_c(k_j)$ has the same distribution for all $k_j$, we will sometimes use $\mathcal{H}_c$ as a shorthand.
    \item For $i \in [m]$, let $S_i \subset [n], |S_i| = n-i$ be a random variable denoting the set of unfilled slots in the array after $i$ keys have been inserted (with the distributed derived from the hashing scheme).
    \item For $i \in [m]$ and $j \in \NN$, let $X_{i, j}$ be a random variable indicating whether the slot indexed by $H_j(k_i)$ is empty at the time that $k_i$ is inserted.
    \item Let $Y_i$ be the position in the probe sequence that $k_i$ uses. In other words, $k_i$ is placed in the slot $H_{Y_i}(k_i)$. We will also use $y_i$ to refer to a (non-random) specific outcome for $Y_i$. Note that the slot must be empty, so $Y_i \in \{r: X_{i, r} = 1\}$ (in a greedy algorithm, the first available slot is taken---in this case, $Y_i=\min \{r: X_{i, r} = 1\}$---but we make no assumptions as to whether the algorithm is greedy or not).
    \item Let $L_i$ be the random variable denoting the location in the array in which the $i$th key is inserted.
\end{itemize}

\subsection{Worst Case Expected \ProbeComplexity}





In this section, we prove the following theorem:

\begin{theorem}
    In any open-addressing scheme achieving load factor $1-\delta$ without reordering, the worst case expected \probecomplexity must be $\Omega(\log \delta^{-1}).$ In particular, there exists some $i \in [m]$ for which $\mathbb{E}[Y_i] = \Omega(\log \delta^{-1})$.
    

    \label{thm:lower}
\end{theorem}

At a high level, we want to reverse the idea of the upper bound algorithms. Rather than partitioning our array into subarrays with exponentially decreasing size, we show that, at minimum, such a construction arises naturally. Given an upper bound $c$ on worst-case expected probe complexity, we will show that there must exist disjoint groups of slots $v_1, v_2, \ldots, v_{\Theta(\log \delta^{-1})}$, of exponentially decreasing sizes, with the property that, for each $i$, $\E[\mathcal{H}_{2c} \cap v_i] \ge \Omega(1)$. This, in turn, implies that $2c \ge \E[|\mathcal{H}_{2c}|] \ge \Omega(\log \delta^{-1})$. As we shall see, the tricky part is defining the $v_i$s in a way that guarantees this property.

\begin{proof}
    Let $c$ be any upper bound for $E[Y_i]$ that holds for all $i$. We want to prove that $c = \Omega(\log \delta^{-1})$.
    Note that, by Markov's inequality, $\Pr [Y_i\leq 2c] \geq \frac{1}{2}$ for any $i \in [m]$.

    Let $\alpha = \left\lfloor\frac{\log \delta^{-1}}{3} \right\rfloor \in\Omega(\log \delta^{-1})$. For $i \in \left[\alpha\right],$ let $a_i = n\left(1-\frac{1}{2^{3i}}\right)$.
    Note that $a_i \leq a_{\log \delta^{-1}/3} \leq n\left(1- \frac{1}{2^{3\log \delta^{-1} /3}}\right) = n(1-\delta) = m.$ Further note that $|S_{a_i}| = n - a_i = \frac{n}{2^{3i}},$ since $S_i$ represents the slots still unfilled; and note that the sizes $|S_{a_i}|$, for $i = 1, 2, \ldots$ form a geometric sequence with ratio $1/2^3 = 1/8$. 
    It follows that, for any $s_{a_i} \leftarrow S_{a_i}$, $s_{a_{i+1}} \leftarrow S_{a_{i+1}}, \ldots$, $s_{a_\alpha}\leftarrow S_{a_\alpha}$,
    even if the $s_{a_i}$'s are \emph{not compatible with each other} (i.e., even if $s_{a_{i + 1}} \not\subseteq s_{a_{i}}$), we have
    $$\left|s_{a_{i+1}} \cup s_{a_{i+2}} \cup \cdots \cup s_{a_\alpha}\right| \leq \sum_{j \ge i + 1} |s_{a_j}| \leq |s_{a_i}| / 7.$$
    Since $\Pr [Y_i \leq 2c] \geq \frac{1}{2}$, we have that, for any $t < 2m-n = n(1-2\delta)$,
    $$\mathbb{E}[|\{i: Y_i \leq 2c \text{ and }  t < i \leq m\}|] \geq \frac{m-t}{2} \geq \frac{n-t}{4} = \frac{|S_t|}{4}.$$
    Therefore, for each $j \in [\alpha-1],$ there is some $s_{a_j} \subseteq[n]$ such that  
    \begin{align}
        \mathbb{E}\left[|\{i:Y_i \leq 2c \text{ and } a_j < i \leq m\}| \, \bigg|  \,S_{a_j} = s_{a_j}\right] \geq \frac{|s_{a_j}|}{4}. \label{eq:sizeMCL}
    \end{align}
    That is, we find some concrete instance $s_{a_j}$ of the random variable $S_{a_j}$ that sets the number of expected ``small'' values of $Y_i$, with $i > a_{j}$ and given $S_{a_j} = s_{a_j}$, to at least the overall expected number. It is important to note that the $s_{a_1}, s_{a_2}, \ldots$ may have an arbitrary relationship to one another; they need not be mutually compatible as values for $S_{a_1}, S_{a_2}, \ldots$. Perhaps surprisingly, even despite this, we will still be able to reason about the relationships between the $s_{a_i}$s. In particular, we will show by the end of the proof that, for each $j$ and for each insertion, the expected number of probes out of the first $2c$ that probe a position in $s_{a_j} \setminus \bigcup_{k > j} s_{a_k}$ is $\Omega(1)$. This will allow us to deduce that, in expectation, the first $2c$ entries the probe sequence contain at least $\Omega(\log \delta^{-1})$ distinct values, implying that $c  = \Omega(\log \delta^{-1})$.
    
    Let 
    \begin{align*}
        \mathcal{L}_j &= \{L_i: m \geq i > a_j \text{ and }  Y_i \leq 2c\} \, \bigg| \, S_{a_j} = s_{a_j}
    \end{align*}
    be the (random variable) set of positions that are used by the ``fast'' insertions (i.e., those satisfying $Y_i \leq 2c$) that take place at times $i > a_j$, given that the set of set of unfilled slots (at time $a_j$) is $s_{a_j}.$ Note that 
    $$\E[|\mathcal{L}_j|] \geq \frac{|s_{a_j}|}{4},$$
    by \eqref{eq:sizeMCL}. Observe that $\mathcal{L}_i \subseteq s_{a_j},$ since all slots filled starting with $s_{a_j}$ as the set of empty slots must come from $s_{a_j}$. We will now argue that, because $\E[|\mathcal{L}_j|]$ is so large, we are guaranteed that $\E[|\mathcal{L}_j \setminus \bigcup_{k > j} s_{a_k}|]$ is also large, namely, $\Omega(|s_{a_j}|)$.
    

    Define
    \begin{align*}
        t_j &= \bigcup_{k > j} s_{a_k},
        \\ v_j &= s_{a_j} \setminus t_j.
    \end{align*}
    and note that the $v_j$s are disjoint:
    
    \begin{claim}\label{claim:disjointvi}
        $v_j \cap v_k = \emptyset$ for all $j \neq k$. That is, all the $v_j$'s are mutually disjoint.
    \end{claim}
    \begin{proof}
        Without loss of generality, suppose $j < k$. By the definition of $t_j,$ $s_{a_k} \subseteq t_j.$ By the definition of $v_k$, $v_k \subseteq s_{a_k}.$ Finally, by the definition of $v_j$, $v_j \cap t_j = \emptyset.$ Therefore, 
        $v_j \cap v_k = \emptyset.$
    \end{proof}

    As we saw earlier, 
    $|t_j| = |s_{a_j+1}\cup \cdots \cup s_{a_{\alpha}}| \leq |s_{a_j}| / 7$. 
    Since $\mathcal{L}_i \subseteq s_{a_j} \subseteq v_j \cup t_j,$ we have that
    \begin{align*}
        \frac{\left|s_{a_j}\right|}{4} \leq \E\left[\left|\mathcal{L}_i\right|\right] = \E\left[\left|\mathcal{L}_i \cap v_j \right|\right] + \E\left[\left|\mathcal{L}_i \cap t_j\right|\right] \leq \E\left[\left|\mathcal{L}_i \cap v_j\right|\right] + \frac{\left|s_{a_j}\right|}{7}.
    \end{align*}
    Subtracting, we get that
    \begin{align}
        \E\left[\left|\mathcal{L}_i \cap v_j\right|\right] \geq \frac{\left|s_{a_j}\right|}{4}-\frac{\left|s_{a_j}\right|}{7} \geq \frac{\left|s_{a_j}\right|}{16}.
    \end{align}
    
    The high-level idea for the rest of the proof is as follows. We want to argue that the $v_j$'s are disjoint sets that each have a reasonably large (i.e., $\Omega(1)$) probability of having an element appear in the first $2c$ probes $\mathcal{H}_{2c}$ of a given probe sequence. From this, we will be able to deduce that $c$ is asymptotically at least as large as the number of $v_j$'s, which is $\Omega(\log \delta^{-1})$. 
    

    Let
    \begin{align*}
        p_{i, j} &= \Pr[Y_i \leq 2c \text{ and } L_i \in v_j],
        \\ q_j &= \Pr[\mathcal{H}_{2c} \cap v_j \neq \emptyset].
    \end{align*}
    We necessarily have $p_{i, j} \leq q_j,$ since, for $Y_i \leq 2c$ and $L_i \in v_j$, we must have that at least some hash function in the first $2c$ outputted an index in $v_j.$ We thus have that
    $$\frac{|s_{a_j}|}{16} \leq \mathbb{E}[|\mathcal{L}_j \cap v_j|] = \sum_{i=a_j+1}^{m} p_{i, j} \leq \sum_{i=a_j+1}^m q_j = q_j(m-a_j) \leq q_j(n-a_j) = q_j |s_{a_j}|.$$
    From this, we conclude that $q_j \geq \frac{1}{16}$ for all $j \in [\alpha-1].$
    Therefore,
    \begin{align*}
        2c = |\{H_i: i \in [2c]\}| &= |\{H_i: i \in [2c]\} \cap [n]| \\
        &= \mathbb{E}[|\{H_i: i \in [2c]\} \cap [n]|] \\
        &\geq \sum_{j=1}^{\alpha-1} \mathbb{E}[|\{H_i: i \in [2c]\} \cap v_j|] \\
        &\geq \sum_{j=1}^{\alpha-1} q_j \geq \frac{1}{16}(\alpha-1) = \Omega(\log \delta^{-1}),
    \end{align*}
    and we are done.
    
\end{proof}

\subsection{High-Probability Worst-Case \ProbeComplexity}
To prove the high probability lower bounds, we will use a similar set construction to that in the previous proof. The one difference is that we now have a cap on the maximum \probecomplexity. In terms of the variables used in the proof of Theorem \ref{thm:lower}, this one extra constraint allows us to obtain a stronger bound on $\E[\mathcal{H}_{2c} \cap v_j]$---namely our bound on this quantity will increase from $\Omega(1)$ to $\Omega(\log \delta^{-1}).$

The main idea is that, since we now have a \emph{worst-case} upper bound $c$ on how many probes an insertion can use, we can more explicitly analyze the actual probability that a particular slot \textit{ever} gets probed. As will show, for a slot to be seen with probability greater than $1-\delta$ (which is necessary for a $(1 - \delta)$-fraction of slots to get filled), it must appear in $\mathcal{H}_c$ with probability at least $\Omega(\log \delta^{-1} / n)$. Integrating this into our analysis, we will be able to pick up an extra $\Omega(\log \delta^{-1})$ factor compared to the proof of Theorem \ref{thm:lower}.


\begin{theorem}\label{thm:highprobdelta}
    In any open-addressing scheme that does not perform reordering, with probability greater than $1/2$, there must be some key whose \probecomplexity ends up as $\Omega(\log^2 \delta^{-1})$. In other words,
    $$\Pr\left[Y_i \leq c\, \forall i \in [m] \right] \leq \frac{1}{2},$$
    for all $c \in o(\log^2 \delta^{-1}).$
\end{theorem}

\begin{proof}



Suppose that some open-addressing scheme that does not perform reordering exists such that, for some $c \in \NN$, the \probecomplexity of all keys is at most $c$ with probability greater than $1/2.$
We will show that $c = \Omega(\log^2 \delta^{-1})$. 

By definition, we have that 
\begin{align}
    \Pr[Y_i \leq c \, \forall i \in [m]] > \frac{1}{2}. \label{eq:HighprobSuccess}
\end{align}
Therefore, for each $i \in \{0, 1, \ldots, m\},$ there must be some $s_i \subseteq [n]$ of size $n-i$ such that 
$$\Pr\left[Y_i \leq c \, \forall i \in [m] \, \bigg| \, S_i = s_i\right] > \frac{1}{2}.$$
Otherwise, by the definition of conditional probability, we would contradict \eqref{eq:HighprobSuccess}.

\begin{claim}\label{claim:probableslots}
    For any $i < n(1-256\delta)$, there must be some set $r_i \subseteq s_i \subseteq [n]$ with $|r_i| > \frac{n-i}{2} = \frac{|s_i|}{2}$ such that, for any $x \in r_i,$
    $$P[x \in \mathcal{H}_c] > \frac{1}{32} \frac{\log\left(\frac{|s_i|}{n\delta}\right)}{|s_i|}.$$
\end{claim}

\begin{proof}
A necessary condition for a table to succeed (that is, for all keys to have a \probecomplexity of at most $c$) with a load factor of $1-\delta$ is that 
$$s_i \setminus \cup_{j > i} \mathcal{H}_c(k_j)$$
has size at most $\delta n$. Indeed, these are the set of slots that are empty after the insertion of $k_i$, and that never get probed by (the first $c$ probes) of any of the remaining insertions.

Therefore,
\begin{align}
    \Pr\left[\left|s_i \cap \left(\bigcup_{j = i+1}^m \mathcal{H}_c(k_j)\right) \right| > |s_i| - \delta n : S_i = s_i\right] > \frac{1}{2}. \label{eq:ProbSuccessSi}
\end{align}

Note that the conditioning on $S_i = s_i$ is unnecessary, as the random variables $\mathcal{H}_c(k_j)$, $j > i$, are independent of the event $S_i = s_i$.

Let $p = \frac{\log\left(\frac{|s_i|}{n\delta}\right)}{|s_i|}$, and let $t_i$ be the set of all slots $x \in s_i$ such that
$$\Pr[x \in \mathcal{H}_c] \leq \frac{p}{32}.$$
We will complete the proof of the claim by showing that $|t_i| < |s_i|/2$. To do so, we will calculate the number of elements in $t_i$ that are expected to appear in some $H_c(k_j)$ with $j > i$:
\begin{align*}
    \mathbb{E}\left[\left|t_i \cap \left(\bigcup_{j = i+1}^m \mathcal{H}_c(k_j)\right) \right|\right] &= \sum_{x \in t_i} \Pr \left[x \in \bigcup_{j = i+1}^m \mathcal{H}_c(k_j)\right]\\
    &= \sum_{x \in t_i} 1 - \left(1 - \Pr \left[x \in \mathcal{H}_c\right]\right )^{m - i} \textrm{ (since $\mathcal{H}_c(k_j)$ is iid across $j$)} \\
    &\leq \sum_{x \in t_i} 1 - \left(1 - \frac{p}{32}\right)^{|s_i| - n\delta} \\
    &\leq \sum_{x \in t_i} 1 - \left(1 - \frac{p}{32}\right)^{|s_i|/2} \tag{since $i < n(1 - 2\delta)$ by assumption}\\
    &\leq |t_i| - |t_i| \left(1 -  \frac{1}{|s_i|}\right)^{\log\left(\frac{|s_i|}{n\delta}\right)|s_i|/64} \textrm{ (since } (1-x/t) \leq (1-1/t)^x \textrm{ if } x, t \geq 1) \\
    &< |t_i| - |t_i| (1/2)^{\log\left(\frac{|s_i|}{n\delta}\right)/8} \\
    &< |t_i| - |t_i| \left(\frac{n\delta}{|s_i|}\right)^{1/8}.
\end{align*}

By assumption, $i < n(1-256\delta)$, so $|s_i| > n - n(1-256\delta) = 256n \delta$. Since $|s_i| > 256n \delta,$ we have that $\left(\frac{n\delta}{|s_i|}\right)^{1/8} < \left(\frac{1}{256}\right)^{1/8} = \frac{1}{2}.$ We thus have that
$$|t_i| - |t_i| \left(\frac{n\delta}{|s_i|}\right)^{1/8} < \frac{|t_i|}{2}.$$
For the table insertion process to succeed, we must have that
$$\mathbb{E}\left[\left|t_i \cap \left(\bigcup_{j = i+1}^m \mathcal{H}_c(k_j)\right) \right|\right] \geq |t_i| - n\delta,$$
as otherwise more than $n\delta$ slots are never part of any hash function output and therefore are guaranteed to be unfilled at the end. It follows that $|t_i| < 2n \delta < |s_i| / 2$, as desired.

\end{proof}

Let $a_i = n\left(1-\frac{1}{4^i}\right)$ for $i \in [\log \delta^{-1} / 4].$ Observe that, for any $i \in [\log \delta^{-1} / 4]$,
$$\left|\bigcup_{j =i+1}^{\log \delta^{-1} / 4} s_{a_j} \right| \leq \sum_{j = 1}^{\log\delta^{-1}/4-i} \frac{\left|s_{a_{i}}\right|}{4^j} \le \frac{|s_{a_i}|}{3} \leq \frac{3 |s_{a_i}|}{8}.$$
Let 
$$v_i = r_{a_i} \setminus \left(\bigcup_{j =i+1}^{\log \delta^{-1} / 4} s_{a_j}\right),$$
for $i \in [\log\delta^{-1}/4]$.
By Claim \ref{claim:probableslots}, we have that \(|r_{a_i}| \geq \frac{|s_{a_i}|}{2},\) so
$$\left|v_i\right| \geq \frac{\left|s_{a_i}\right|}{2} - \frac{3\left|s_{a_i}\right|}{8} = \frac{\left|s_{a_i}\right|}{8}.$$

Note that $v_i \cap v_j = \emptyset$ for all $i \neq j$ with a similar proof to Claim \ref{claim:disjointvi} in the previous subsection. Also, note that 
$$|s_{a_i}| \geq n\left(\frac{1}{4}\right)^{\log \delta^{-1}/4} = n\left(\frac{1}{2}\right)^{\log \delta^{-1}/2} = n \sqrt{\delta}.$$

We now obtain a lower bound on $|\mathcal{H}_c| \le c$ by unrolling the definition of each $v_i$. In particular, assuming without loss of generality that $\log \delta^{-1}/4 > 256,$ we have 
\begin{align*}
    c \geq \mathbb{E}[|\mathcal{H}_c|] &\geq \sum_{i = 1}^{\log \delta^{-1}/4} \mathbb{E}[|\mathcal{H}_c \cap v_i|] \\
    &= \sum_{i = 1}^{\log \delta^{-1}/4} \sum_{x \in v_i} \mathbb{E}[|\{x\} \cap \mathcal{H}_c|] \\
    &= \sum_{i = 1}^{\log \delta^{-1}/4} \sum_{x \in v_i} \Pr[x \in \mathcal{H}_c] \\
    &\geq \sum_{i = 1}^{\log \delta^{-1}/4} \sum_{x \in v_i} \frac{1}{32}\frac{\log\left(\frac{|s_{a_i}|}{n\delta}\right)}{|s_{a_i}|} \\
    &= \frac{1}{32} \sum_{i = 1}^{\log \delta^{-1}/4} |v_i| \frac{\log\left(\frac{|s_{a_i}|}{n\delta}\right)}{|s_{a_i}|}\\
    &\geq \frac{1}{32} \sum_{i = 1}^{\log \delta^{-1}/4} \frac{|s_{a_i}|}{8} \frac{\log\left(\frac{n\sqrt{\delta}}{n\delta}\right)}{|s_{a_i}|} \\
    &= \frac{1}{32} \sum_{i = 1}^{\log \delta^{-1}/4} \frac{\log \delta^{-1}}{16} \tag{since $\log (1/\sqrt{\delta}) = \log \delta^{-1}/2$}\\
    &=\frac{1}{32} \cdot \frac{1}{16} \cdot \frac{1}{4} \log^2 \delta^{-1} = \Omega(\log^2 \delta^{-1}),
\end{align*}
as desired.


\end{proof}


We now combine our result above with a known result to obtain our full lower bound:

\begin{theorem}
    In any open-addressing scheme that does not support reordering, there must be some key whose \probecomplexity ends up as $\Omega(\log \log n + \log^2 \delta^{-1})$ with probability greater than $1/2$, assuming that $1-\delta = \Omega(1).$
\end{theorem}

\begin{proof}
    We only need to prove that some key has \probecomplexity $\Omega(\log \log n),$ as we already proved there is some key with \probecomplexity $\Omega(\log^2 \delta^{-1})$ in Theorem \ref{thm:highprobdelta}. Our proof mirrors the proof of Theorem 5.2 in \cite{bender2023tinypointers}, which in turn is primarily based on the following theorem in \cite{vocking2003asymmetry}:
    \begin{theorem}[Theorem 2 in \cite{vocking2003asymmetry}]\label{thm:dchoicelowerbound}
        Suppose $m$ balls are sequentially placed into $m$ bins using an arbitrary mechanism, with the sole restriction that each ball chooses between $d$ bins according to some arbitrary distribution on $[m]^d$. Then the fullest bin has $\Omega(\log \log n / d)$ balls at the end of the process with high probability.
    \end{theorem}

    Now, suppose we have some arbitrary open-addressing scheme that does not perform reordering, where, with probability greater than $1/2$, all keys have \probecomplexity at most $d$. We modify our hash table scheme into a balls and bins process that chooses between at most $d$ bins as follows.

    Suppose key $k_i$ is inserted into location $l_i = h_j(k_i).$ If $j \leq d$, then place ball $i$ into bin $h_j(k_i) \textrm{ mod } m$. Otherwise, place ball $i$ into bin $h_{d} (k_i) \textrm{ mod } m,$ in this way ensuring that the scheme chooses between at most $d$ bins; the set of possible choices is $\{H_j(k_i) \textrm{ mod } m: j \leq d\}$, a set of size (at most) d. This process also ensures that the fullest bin most likely has few balls land in it:

    \begin{lemma}\label{lemma:fullestbin}
        With probability greater than $1/2$, the fullest bin has $O(1)$ balls at the end of the process.
    \end{lemma}

    \begin{proof}
        Suppose that ball $i$ lands in bin $B_i$. The indices of the balls that land in the $j$th bin are thus
        $\{i: B_i = j\}.$
        
        Now, suppose that $B_i = L_i \text { mod } m$ for all $i \in [m],$ and note that this happens with probability greater than $1/2.$
        Since each slot can only store one key, $L_i \neq L_j$ for any $i \neq j$. Therefore, $\{i: B_i = j\} \subseteq \{i \in [n]: i \text{ mod } m = j\}.$ Since $1-\delta = \Omega(1),$ we have that $m = \Omega(n),$ or $n = O(m).$ Thus, $|\{i \in [n]: i \text{ mod } m = j\}| = O(1)$ for all $i \in [m],$ and the fullest bin has at most $O(1)$ balls, as desired.
    \end{proof}
    

    By Theorem \ref{thm:dchoicelowerbound}, the fullest bin has $\Omega(\log \log n / d)$ balls at the end of the process with high probability. 
    If $d = o(\log \log n),$ then the fullest bin has $\omega(1)$ balls in the fullest bin with high probability, contradicting Lemma \ref{lemma:fullestbin}. Therefore, $d = \Omega(\log \log n),$ as desired.
\end{proof}

\section{Acknowledgments and Funding}
The authors would like to thank Mikkel Thorup for several useful conversations, including feedback on an earlier version of this paper.

William Kuszmaul was partially supported by a Harvard Rabin Postdoctoral Fellowship and by a Harvard FODSI fellowship under NSF grant DMS-2023528.  Mart\'{\i}n Farach-Colton was partially supported by the Leanard J. Shustek Professorship and NSF grants CCF-2106999, NSF-BSF CCF-2247576, CCF-2423105, and CCF-2420942.

\bibliographystyle{plainurl}
\bibliography{bib}

\end{document}